%% file: main.tex
\documentclass[submission,copyright,creativecommons]{eptcs}

\input{mainmac}
\input{pictures}

\title{Path-Based Program Repair}
\author{Heinz Riener$^1$
\institute{$^1$Institute of Computer Science\\
University of Bremen, Germany}
\email{\{hriener,rehlers,fey\}@cs.uni-bremen.de$\!\!$}
\and R{\"u}diger Ehlers$^{1,2}$
\institute{$^2$DFKI GmbH\\
Bremen, Germany}
\and
G{\"o}rschwin Fey$^{1,3}$
\institute{$^3$Institute of Space Systems,\\
German Aerospace Center, Germany}
\email{goerschwin.fey@dlr.de}
}
\def\titlerunning{Path-Based Program Repair}
\def\authorrunning{H. Riener, R. Ehlers, and G. Fey}
\begin{document}
\maketitle

\begin{abstract}
We propose a path-based approach to program repair for imperative
programs. Our repair framework takes as input a faulty program, a
logic specification that is refuted, and a hint where the fault may be
located.  An iterative abstraction refinement loop is then used to
repair the program: in each iteration, the faulty program part is
re-synthesized considering a symbolic counterexample, where the
control-flow is kept concrete but the data-flow is symbolic.  The
appeal of the idea is two-fold: 1)~the approach lazily considers
candidate repairs and 2)~the repairs are directly derived from the
logic specification. In contrast to prior work, our approach is
complete for programs with finitely many control-flow paths, i.e., the
program is repaired if and only if it can be repaired at the specified
fault location.  Initial results for small programs indicate that the
approach is useful for debugging programs in practice.
\end{abstract}

\section{Introduction}
Debugging is one of the most frequent and challenging activities in
software development.  In order to fix a faulty program without
introducing new bugs or subtle corner cases, an in-depth understanding
of the source code is required.  Error hints and counterexamples
produced by static analysis or model checking tools are only of little
help: they typically report a symptom of a failure but do not point to
the actual cause or provide a repair for a program.  Manually
correcting a faulty program based on this information is hard and
often becomes an iterative trial-and-error process driven by a
developer's intuition.

Automatic program repair techniques aim at reducing this manual burden
by utilizing a refuted logic specification to automatically compute a
repair for a faulty program.  Existing approaches mainly suffer from
two problems: 1)~they do not scale well to large programs or 2)~they
rely on structural restrictions of the considered repairs.  The two
problems are orthogonal.  Scalability issues originate from the fact
that correctness of the entire program is modeled.  Structural
restrictions allow for the enumeration of potential repairs and make
the repairs more readable \cite{Koenig11}.  The choice of the
``right'' structure for a repair, however, is left to the user or
guided by a brute-force search.  Bad choices cause an exclusion of
suitable repairs.

In this paper, we take program repair one step further and address
these two problems with a novel approach, called \emph{path-based
program repair}.  The approach combines symbolic path reasoning and
software synthesis.  We present a repair framework to automatically
correct a faulty imperative program based on an iterative abstraction
refinement loop, assuming that the location of the fault is known.
For instance, the location may be guessed by a user or computed using
a fault localization approach, e.g.,~\cite{Jose11,Riener12}.

In each iteration of the loop, a model checker computes a symbolic
counterexample.  Symbolic counterexamples keep the data-flow symbolic
but the control-flow concrete.  Symbolic path reasoning then infers
verification conditions in the local context of the faulty program
part.  The inferred verification conditions are used to formulate a
synthesis problem.  The faulty program part is then re-synthesized in
order to exclude all symbolic counterexamples found so far.
Consequently, the synthesized candidate repair corrects the program
with respect to the considered symbolic counterexamples.  The loop
terminates when no further symbolic counterexamples are found and is
guaranteed to converge for programs with finitely many control-flow
paths (provided that the model checker and the synthesis procedure do
not diverge).

The appeal of the idea is two-fold: 1)~our approach lazily repairs a
faulty program by inferring verification conditions from
counterexamples utilizing symbolic path reasoning and 2)~the repairs
are derived from the refuted logic specification using software
synthesis.  Symbolic path reasoning, on the one hand, has been
effectively used to cope with scalability.  Software synthesis, on the
other hand, does not restrict candidate repairs to have a predefined
structure~\cite{Koenig11,Bloem12}, such as linear expressions over
program variables.

Path-based program repair establishes a framework for fixing faulty
programs, building on existing ideas, but tunes them to specifically
address program repair.  Na\"ively combining prior work is
problematic: applying predicate abstraction~\cite{Flanag02} leads to
coarse abstractions that are just good enough to either verify the
program or refute its correctness along one of its executions.  Thus,
the abstraction is insufficient to subsequently find a ``good'' repair
for the program, which should work for most, if not all, cases.
Recent software synthesis approaches, e.g., those for the
\emph{Syntax-Guided Synthesis}~(SyGuS)~\cite{Alur13} problem, cannot
decide realizability, which is crucial for termination in case the
initially given fault locations cannot be repaired.

\textbf{Contribution.} The contribution of the paper can be summarized as follows:
\begin{enumerate}
\item an iterative abstraction refinement approach for program repair
combining symbolic path reasoning and software synthesis, called
path-based program repair;
\item a prototype implementation of the repair framework utilizing
domain finitizing and synthesis based on {\em Binary Decision
Diagrams} (BDD), which allows deciding realizability;
\item initial experimental results for our prototype on a small ANSI-C
program.
\end{enumerate}

The remainder of the paper is structured as follows: in
\Sec~\ref{sec:bg}, we describe the necessary background and in
\Sec~\ref{sec:pbrepair}, we present path-based program repair.
\Sec~\ref{sec:eval} is dedicated to our prototype implementation
utilizing domain finitizing and BDD-based synthesis. We also give
experimental results for a small ANSI-C program.
\Sec~\ref{sec:conclusion} concludes the paper.

\section{Background}\label{sec:bg}
\subsection{Program and Specification}\label{ssec:program}
We focus on sequential, finite-state systems described in a high-level
programming language like ANSI-C.  Let~$\Stmts$ be the set of all
statements, a program corresponds to a finite-state automaton
over~$\Stmts$, called \emph{program automaton}, with control locations
as nodes and program operations as edges.  Without loss of generality,
we assume that each program automaton has a distinguished entry node
and a distinguished exit node denoting the program's entry and the
program's exit.

Our approach to program repair is dedicated to static analysis of
individual \emph{control-flow paths} of a program.  A control-flow
path is a consecutive sequence of nodes starting at the entry node and
ending at the exit node, i.e., a word in the language of the program
automaton, which corresponds to a terminating execution of the program
and respects the semantics of the statements.  Loops are unrolled and
decisions at branching points in the program are modeled, e.g.,
\texttt{if} or \texttt{while} statements are replaced with respective
assumption statements.  More formally, a control-flow path $\pi = s_1
s_2 \cdots s_n$ is a sequence of side-effect free statements given in
\emph{Static Single Assignment}~(SSA) form, where each~$s_i$ is either
an \emph{assumption statement} or an \emph{assignment statement}.
Along a fixed control-flow path, computing SSA form is straightforward
because the costly placement of $\phi$-nodes to assemble the
control-flow from different control-flow paths is not necessary.  An
assignment statement is of form \mbox{\tt v := e}, where~$v$ is a
program variable and~$e$ is an expression over program variables and
constants.  An assumption statement is of form \mbox{\tt assume(c)},
where~$c$ is a condition over the program variables and constants.  We
assume that the concatenation operation $\cdot$ is defined for a
control-flow path $\pi = s_1 s_2 \dots s_n$ in the usual form such
that the control-flow path can be represented as $\pi = \pi_A \cdot
s_n$, $\pi = s_1 \cdot \pi_A$, or $\pi = \pi_A \cdot \pi_B$.

A (program) state~$\sigma$ is a valuation of all program variables,
i.e., a mapping from the program variables to values within the
respective domains.  The specification is given as a precondition and
a postcondition.  The precondition defines the set of program states
initially possible at the program's entry.  The postcondition defines
the set of program states allowed at the program's exit.  We say that
a program~$P$ is \emph{correct}~\emph{if and only if}~(iff) all
executions starting in a state~$\sigma$ specified by the precondition
reach a state~$\sigma^\prime$ on termination that fulfills the
postcondition.  Otherwise we say that $P$ is \emph{faulty}.

In the following, sets of states are symbolically represented
utilizing \emph{First-Order Logic}~(\FOL) formul\ae{}.  Let~$\phi \in
\FOL$, we write~$\Vars(\phi)$ and~$\FV(\phi)$ to denote the variables
and the free variables of~$\phi$, respectively.  In several parts of
the descriptions, we explicitly show that a \FOL formula~$\phi$
depends on a specific variable~$v$ by writing $\phi[v]$ and generalize
this notation to show the dependence on a list~$\overline{v}$ of
variables by writing $\phi[\overline{v}]$.  As usual, we call a \FOL
formula $\phi$ \emph{satisfiable} iff an assignment to $Vars(\phi)$
exists that makes the formula~\TRUE and \emph{unsatisfiable}
otherwise.  Moreover, we say a \FOL formula $\phi$ is \emph{valid} iff
$\phi$ is equivalent to \TRUE, i.e., $\phi$ is \TRUE for all
assignments to $\Vars(\phi)$.  Lastly,~$\phi[x/y]$ denotes~$\phi$
where each free occurrence of variable~$x$ is replaced by
variable~$y$.

\subsection{Symbolic Path Reasoning}\label{ssec:propagation}
\begin{figure}[t]
  \centering
  \scriptsize
  \begin{minipage}{0.49\columnwidth}
    \centering
    \begin{Verbatim}[commandchars=\\\{\},gobble=6,codes={\catcode`\$=3\catcode`\^=7\catcode`\_=8}]
      \PY{c+c1}{// $\mathbf{\varphi :\Leftrightarrow \{x=0\}}$}
      \PY{l+m+mf}{1.} \PY{k}{assume}\PY{p}{(}\PY{n}{x} \PY{o}{\PYZlt{}} \PY{l+m+mi}{2}\PY{p}{)}\PY{p}{;}
      \PY{c+c1}{// $\{x=0 \wedge x<2\} \Leftrightarrow \{x=0\}$}
      \PY{l+m+mf}{2.} \PY{n}{x} \PY{o}{=} \PY{n}{x} \PY{o}{+} \PY{l+m+mi}{1}\PY{p}{;}
      \PY{c+c1}{// $\{x=1\}$}
      \PY{l+m+mf}{3.} \PY{k}{assume}\PY{p}{(}\PY{n}{x} \PY{o}{\PYZlt{}} \PY{l+m+mi}{2}\PY{p}{)}\PY{p}{;}
      \PY{c+c1}{// $\{x=1 \wedge x<2\} \Leftrightarrow \{x=1\}$}
      \PY{l+m+mf}{4.} \PY{n}{x} \PY{o}{=} \PY{n}{x} \PY{o}{+} \PY{l+m+mi}{1}\PY{p}{;}
      \PY{c+c1}{// $\{x=2\}$}
      \PY{l+m+mf}{5.} \PY{k}{assume}\PY{p}{(}\PY{o}{!}\PY{p}{(}\PY{n}{x} \PY{o}{\PYZlt{}} \PY{l+m+mi}{2}\PY{p}{)}\PY{p}{)}\PY{p}{;}
      \PY{c+c1}{// $\{x=2 \wedge \neg(x < 2)\} \Leftrightarrow \{x=2\}$}
    \end{Verbatim}
  \end{minipage}
  \begin{minipage}{0.49\columnwidth}
    \centering
    \begin{Verbatim}[commandchars=\\\{\},gobble=6,codes={\catcode`\$=3\catcode`\^=7\catcode`\_=8}]
      \PY{c+c1}{// $\{true\}$}
      \PY{l+m+mf}{1.} \PY{k}{assume}\PY{p}{(}\PY{n}{x} \PY{o}{\PYZlt{}} \PY{l+m+mi}{2}\PY{p}{)}\PY{p}{;}
      \PY{c+c1}{// $\{true\}$}
      \PY{l+m+mf}{2.} \PY{n}{x} \PY{o}{=} \PY{n}{x} \PY{o}{+} \PY{l+m+mi}{1}\PY{p}{;}
      \PY{c+c1}{// $\{(x<2) \rightarrow x\leq{}1\} \Leftrightarrow \{true\}$}
      \PY{l+m+mf}{3.} \PY{k}{assume}\PY{p}{(}\PY{n}{x} \PY{o}{\PYZlt{}} \PY{l+m+mi}{2}\PY{p}{)}\PY{p}{;}
      \PY{c+c1}{// $\{x+1 \leq 2\} \Leftrightarrow \{x\leq{}1\}$}
      \PY{l+m+mf}{4.} \PY{n}{x} \PY{o}{=} \PY{n}{x} \PY{o}{+} \PY{l+m+mi}{1}\PY{p}{;}
      \PY{c+c1}{// $\{\neg(x<2) \rightarrow x=2\} \Leftrightarrow \{x\leq{}2\}$}
      \PY{l+m+mf}{5.} \PY{k}{assume}\PY{p}{(}\PY{o}{!}\PY{p}{(}\PY{n}{x} \PY{o}{\PYZlt{}} \PY{l+m+mi}{2}\PY{p}{)}\PY{p}{)}\PY{p}{;}
      \PY{c+c1}{// $\mathbf{\psi :\Leftrightarrow \{x=2\}}$}
    \end{Verbatim}
  \end{minipage}
  \caption{Predicate propagation along a control-flow path.}
  \label{Fig:Propagation}
\end{figure}

We use Floyd/Hoare style computation to propagate predicates along a
control-flow path.  This is done by applying standard predicate
transformers like weakest precondition and strongest postcondition to
control-flow paths.  A predicate transformer is a function
$\mathop{pt} : \FOL \times \Stmts \rightarrow \FOL$ that maps a
formula $\phi \in \FOL$ and a statement $s \in \Stmts$ to a formula
$\mathop{pt}(\phi,s) \in \FOL$.

\begin{defn}
Given a statement $s \in \Stmts$ executed on state $\sigma$ to produce
state $\sigma^\prime$ and a predicate $\phi \in \FOL$, the weakest
precondition transformer $wp$ and the strongest postcondition
transformer $sp$ compute
\begin{itemize}
\item the weakest predicate~$wp(\phi,s) \in \FOL$, called \emph{weakest
precondition}, that guarantees \mbox{$\sigma \models wp(\phi,s)$} if
\mbox{$\sigma^\prime \in \phi$} and
\item the strongest predicate~$sp(\phi,s) \in \FOL$, called
\emph{strongest postcondition}, that guarantees \mbox{$\sigma^\prime
\models sp(\phi,s)$} if \mbox{$\sigma \in \phi$}, respectively.
\end{itemize}
\end{defn}

Let $\phi \in \FOL$ and $s \in \Stmts$, we define the usual mechanic
rules to compute the weakest precondition and the strongest
postcondition for the two types of statements that occur in a
control-flow path, respectively,
\begin{align*}
wp(\phi,\mbox{assume(c)}) &\Leftrightarrow c \rightarrow \phi\\
wp(\phi,\mbox{v := e}) &\Leftrightarrow \forall v^\prime . (v^\prime = e \rightarrow \phi[v/v^\prime]) \Leftrightarrow \phi[v/e]\\
sp(\phi,\mbox{assume(c)}) &\Leftrightarrow c \wedge \phi\\
sp(\phi,\mbox{v := e}) &\Leftrightarrow \exists v^\prime . (v = e[v/v^\prime] \wedge \phi[v/v^\prime]),
\end{align*}
and naturally generalize them to sequences $s_1 s_2 \cdots s_n$ of statements,
\begin{align*}
wp(\phi,s_1 s_2 \dots s_n) &\Leftrightarrow wp(wp(\phi,s_n), s_1 s_2 \cdots s_{n-1})\\
sp(\phi,s_1 s_2 \dots s_n) &\Leftrightarrow sp(sp(\phi,s_1), s_2 s_3 \cdots s_n).
\end{align*}

Since loops are already unrolled, no loop invariants have to be found,
and thus computing strongest postconditions and weakest preconditions
along control-flow paths is decidable if the logic in use admits
quantifier elimination.  Moreover, note that computing the weakest
precondition of an assignment statement amounts to replacing a
variable by an expression. Thus, the application of costly quantifier
elimination procedures is not necessary in this case.

We refer to the application of $wp$ and $sp$ to a given control-flow
path $\pi = s_1s_2 \dots s_n$, i.e., a sequence of statements that
respect the semantics of a program, and to a \FOL formula as
\emph{backward propagation} and
\emph{forward propagation}, respectively.

\begin{ex}
In Fig.~\ref{Fig:Propagation}, we apply forward propagation and
backward propagation, respectively, to a control-flow path~$s_1 s_2
\dots s_5$.  Forward propagation (on the left) uses the
precondition~$\varphi \Leftrightarrow \{x=0\}$ and the backward
propagation (on the right) uses the postcondition~$\psi
\Leftrightarrow \{x=2 \}$.
\end{ex}

\begin{defn}\label{def:hoare}
Let~$\pi$ be a control-flow path with precondition~$\varphi$ and
postcondition~$\psi$.  We say that a Hoare triple~$\{\varphi\} \pi
\{\psi\}$ holds iff the two equal conditions $\varphi \rightarrow
wp(\psi,\pi)$ and $sp(\varphi,\pi) \rightarrow \psi$ are valid.
Otherwise, we call $\pi$ a (symbolic) \emph{counterexample} for
$\varphi$ and $\psi$.
\end{defn}

This definition differs from the standard connotation of a
counterexample given by a concrete input assignment returned by a
model checking procedure.  However, a symbolic counterexample in our
sense can directly be determined after model checking by interpreting
the program on the concrete input assignment and logging the
corresponding control-flow path; thus, control-flow is concrete but
data-flow symbolic.

\subsection{Software Synthesis}\label{ssec:synthesis}
We use synthesis to repair faulty programs and treat a synthesis
procedure as a black box that derives program terms from a logic
specification.  The logic specification is a predicate
$\phi[\overline{x},\overline{y}] \in \FOL$, where $\overline{x}$ is a
set of uncontrollable variables, $\overline{y}$ is a set of
controllable variables, and $\FV(\phi) \subseteq \overline{x} \cup
\overline{y}$.
A synthesis procedure computes terms over $\overline{x}$ such that
replacing all occurrences of $\overline{y}$ in $\phi$ by their
respective term yields a valid formula.

\begin{defn}\label{def:synthesis}
A (software) synthesis procedure computes terms~$\overline{T}$ from a
given predicate~$\phi \in \BV$ and variables $\overline{x} \subseteq
\FV(\phi)$ such that
\begin{align}
\forall \overline{x} . \phi[\overline{y}/\overline{T}]
\end{align}
is valid, where $\overline{y} = FV(\phi) \backslash \overline{x}$.

A synthesis procedure may choose not to compute a term or may not
terminate.  If this can only happen if such a term replacement does
not exist, i.e., if
\begin{equation*}
\forall \overline{x} \exists \overline{y} . \phi
\end{equation*}
is valid,
we call a synthesis procedure \emph{complete}. We also say
that a complete synthesis procedure which always terminates is able to
\emph{detect unrealizability}.
\end{defn}

Note that this formulation defines a \emph{software synthesis}
problem~\cite{Kuncak10,Hamza10}, where we search for a piece of
terminating code that satisfies a given logic specification. This is
in contrast to works on \emph{reactive synthesis}, where a
finite-state machine is to be synthesized that executes for an
unbounded duration of time and satisfies a specification in some
temporal logic, i.e., that reasons about the behavior of the
finite-state machine over time.

Synthesis corresponds to quantifier elimination and synthesis
procedures have been provided for relations expressed in different
decidable logics, e.g., Boolean logic, linear arithmetic and
sets~\cite{Kuncak13}, unbounded bit-vectors~\cite{Hamza10}, term
algebras and the theory of integer-indexed arrays with symbolic bounds
on index ranges~\cite{Jacobs13}.

\section{Path-Based Program Repair}\label{sec:pbrepair}

\subsection{The \PBRepair{} Framework}
The overall \PBRepair{} approach is described as pseudo code shown in
\Alg~\ref{alg:pbrepair}.  Additionally, \Tab~\ref{tab:pbrepair} gives
a description of the main components used in the algorithm, and
\Fig~\ref{fig:pbrepair} describes the interaction between them.

The input to \PBRepair{} is a faulty program~$P$, a logic
specification given as a pair of a precondition and a postcondition
$(\varphi,\psi)$, and an fault region~$e$ to be repaired, where $e$
contains assignment statements to a set of
variables~$\overline{v}$. The algorithm returns a repaired
program~$P^\prime$ which is a copy of~$P$ but in which the code within
the fault region~$e$ has been replaced by assignment statements to the
variables~$\overline{v}$ such that~$P^\prime$ is correct with respect
to~$\varphi$ and~$\psi$.

The algorithm can be seen as an iterative abstraction refinement loop
guided by the counterexamples provided by a model checker.  The
algorithm maintains a set~$\Pi$ of counterexamples and modifies a
copy~$P^\prime$ of the program (line~1).  In each iteration, three
steps are performed: firstly, the program is model checked with
respect to its logic specification (line~2).  If verification
succeeds, then the algorithm terminates with the currently considered
program~$P^\prime$ as output.  Otherwise, a (symbolic)
counterexample~$\pi$ is provided by the model checker and added
to~$\Pi$ (line~3).

Secondly, a synthesis procedure is invoked with the predicate~$\Phi
\in \FOL$ and a set of variables~$\overline{y} \subseteq Var(\phi)$ to
be synthesized (line~5).  The predicate~$\Phi \in \BV$ accumulates the
verification conditions by propagating the precondition~$\varphi$ and
the postcondition~$\psi$ to the local context of the fault region~$e$
for all counterexamples $\pi~\in~\Pi$~(line 4).  For a counterexample
$\pi= \pi_A \cdot e \cdot \pi_B$ we use the strongest postcondition
transformer~$\mathit{sp}$ to propagate~$\varphi$ forward and the
weakest precondition transformer~$\mathit{wp}$ to propagate~$\psi$
backward until reaching the fault region.  The variables
~$\overline{y}$ are fresh variables replacing~$\overline{v}$
in~$\mathit{wp}(\psi,\pi_B)$ and otherwise do not occur in~$\Phi$.
The terms~$\overline{T}$ produced as a result of synthesis are a
repair for the fault region~$e$ in program~$P$ considering all
counterexamples in~$\Pi$.  If $\Phi$ is unrealizable~(line 6-8), then
the algorithm terminates with an error indicating that the program
cannot be repaired in the considered fault region~$e$.  For instance,
this may happen if the initial fault region~$e$ has been provided by a
user or an unsound fault localization algorithm.  Otherwise, synthesis
yields a list of terms~$\overline{T}$ that are used to repair the
program.

Thirdly, if the terms~$\overline{T}$ could be computed~(line~5),
assignment statements to the variables in~$\overline{v}$ are generated
as a repair to the program in the syntax of the programming language
in use. In~$P^\prime$, the assignment statements replace the
statements within the fault region~(line 9).  The program~$P^\prime$
is correct (by construction) with respect to $\varphi$, $\psi$, and
all counterexamples $\pi \in \Pi$.  The algorithm loops until a
correctly repaired program is found or the non-existence of a repair
is detected.

\begin{algorithm}[t]
\DontPrintSemicolon
\SetKwInOut{Input}{input}
\SetKwInOut{Output}{output}
\Input{faulty program~$P$, precondition~$\varphi$, postcondition~$\psi$, fault region~$e$}
\Output{repaired program~$P^\prime$, such that $P^\prime$ is correct}
\SetKwFunction{ModelCheck}{\textsf{ModelCheck}}
\SetKwFunction{Synthesize}{\textsf{Synthesize}}
\SetKwFunction{ApplyRepair}{\textsf{ApplyRepair}}
\SetKw{throw}{throw}
\SetKwIF{Let}{ThenLet}{Else}{let}{}{error}{error}{}
\SetKwFunction{pre}{$\mathsf{Pre}_e$}
\SetKwFunction{post}{$\mathsf{Post}_e$}
\SetKwProg{subf}{fun}{}{end}
$\Pi := \emptyset$, $P^\prime := P$\;
\While{$\pi := \ModelCheck(P^\prime,\varphi,\psi)$}{
  $\Pi := \Pi \cup \{\pi\}$\;
  \Let{$\Phi := \bigwedge\limits_{\pi_A\cdot e \cdot \pi_B\in\Pi} (\mathsf{sp}(\varphi,\pi_A) \Rightarrow (\mathsf{wp}(\psi,\pi_B)[\overline{v}/\overline{y}]))
$}{
    $\overline{T} := \Synthesize(\Phi,\overline{y})$\;
    \If{\textnormal{Unrealizable}}{\textnormal{\throw{Unable to repair in fault region~$e$.}}}
  }
  $P^\prime := \ApplyRepair(P^\prime,e,\overline{T})$\;
}
\Return $P^\prime$\;
\caption{\PBRepair}
\label{alg:pbrepair}
\end{algorithm}

\begin{figure}[t]
\centering
\includegraphics[scale=0.38]{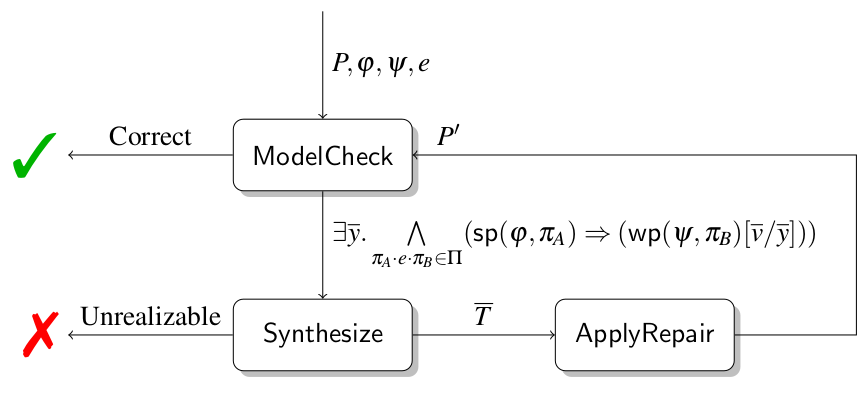}
\caption{\PBRepair{}: high-level overview of path-based program repair.}
\label{fig:pbrepair}
\end{figure}

\begin{table}
\centering
\begin{tabular}{lp{13.3cm}}
\toprule
\bf Name & \textsf{ModelCheck} \\
\midrule
\it Input & Source code of a program~$P$, a precondition~$\varphi$,
and a postcondition~$\psi$. \\
\it Output & A counterexample if~$P$ is faulty and \FALSE otherwise.\\
  \it Description & Model checks the source code of~$P$ assuming~$\varphi$ and asserting~$\psi$ and returns a (symbolic)
counterexample $\pi$.\\
\midrule
\bf Name & \textsf{Synthesize} \\
\midrule
\it Input & A predicate $\phi$ and a vector $\overline{x} \subseteq \FV(\phi)$ of controllable variables.\\
\it Output & A vector $\overline{T}$ of terms if $\phi$ is realizable with respect to $\phi$ or throws an exception.\\
\it Description & Computes a vector $\overline{T}$ of program terms on termination, such that $\forall \overline{x} \exists \overline{y} . \phi \rightarrow \forall \overline{x} . \phi[\overline{y}/\overline{T}]$ becomes valid.  If $\forall \overline{x} \exists \overline{y} . \phi$ is not valid, i.e., is unsatisfiable, an exception is thrown which indicates that the synthesis problem is unrealizable.\\
\midrule
\bf Name & \textsf{ApplyRepair} \\
\midrule
\it Input & Source code of program~$P$, a fault region~$e$, and program terms~$\overline{T}$. \\
\it Output & A repaired copy $P^\prime$ of program~$P$, where all changes are applied to fault region~$e$. \\
\it Description & Transforms the program terms~$\overline{T}$ into source code in the
programming language in use and replaces them for the right-hand side of the assignment statements in fault region~$e$ in $P^\prime$.\\
\bottomrule
\end{tabular}
\caption{A description of the main components used in \Alg~\ref{alg:pbrepair}.}
\label{tab:pbrepair}
\end{table}

\subsection{Correctness and Termination}

\begin{lem}\label{ref:iter}
Let~$\pi = \pi_A \cdot e \cdot \pi_B$ be a counterexample to
precondition~$\varphi$ and postcondition~$\psi$, where $e$ is a list
of assignment statements ${\tt v_i := e_i}$ that is known to be faulty
and $\pi_A$ and $\pi_B$ are known to be correct, then $\{\varphi\}
\pi_A \cdot r \cdot \pi_B \{\psi\}$ holds for the assignment
statements~$r$ of form $r:{\tt v_1 := T_1 \ldots v_n := T_n}$, where
$\overline{v} = [v_1, \ldots, v_n]$ and $\overline{T} = [T_1, \ldots,
T_n]$ is the result of performing synthesis on the specification $
sp(\varphi,\pi_{A}) \rightarrow
wp(\psi,\pi_{B})[\overline{v}/\overline{y}]$ for the output variables
$\overline{y}$.
\end{lem}

\begin{proof}
The control-flow path $\pi = \pi_A \cdot r \cdot \pi_B$ is correct
with respect to precondition~$\varphi$ and postcondition~$\psi$ iff
$\{\varphi\} \pi_A \cdot r \cdot \pi_B \{\psi\}$ holds.  Let
$\Gamma_1, \Gamma_2 \in \FOL$ be the verification conditions that hold
immediately before and immediately after the execution of the
statements~$r$, respectively, such that $\{\varphi\} \pi_A \cdot r
\cdot \pi_B \{\psi\}$ decomposes to the three Hoare triples
$H_1:~\{\varphi\} \pi_A \{\Gamma_1\}$, $H_2:~\{\Gamma_2\} \pi_B
\{\psi\}$, and $H_3:~\{\Gamma_1\} r \{\Gamma_2\}$.  Since $\pi_A$ and
$\pi_B$ are correct, $H_1$ and $H_2$ hold, and only $H_3$ has to be
proven.  The Hoare triple $H_3$ holds iff (i)~$\Gamma_1 \rightarrow
wp(\Gamma_2, r:{\tt v_1 := T_1 \ldots v_n := T_n})$ is valid
(Def~\ref{def:hoare}), which further simplifies to $\Gamma_1
\rightarrow \Gamma_2[\overline{v}/\overline{T}]$.  The Hoare triples
$H_1$ and $H_2$ hold iff (ii)~$sp(\varphi,\pi_A) \rightarrow \Gamma_1$
and (iii)~$\Gamma_2 \rightarrow wp(\psi,\pi_B)$ hold
(Def~\ref{def:hoare}), respectively.  We use (ii) to weaken the
left-hand side of (i) and (iii) to strengthen the right-hand side of
(i), and obtain (iv) $sp(\varphi,\pi_A) \rightarrow
wp(\psi,\pi_B)[\overline{v}/\overline{T}]$.  The term
replacements~$\overline{T}$ guarantee the validity of implication (iv)
(Def~\ref{def:synthesis}), and thus $H_3$ and $\{\varphi\} \pi
\{\psi\}$ hold.
\end{proof}

\begin{thm}
Let~$P$ be a program with finitely many control-flow paths,
$(\varphi,\psi)$ be a pair of a precondition and a postcondition, and
$e: {\tt v_1 = w_1; \ldots v_n := w_n}$ be a fault region in $P$, then
algorithm~$\PBRepair(P, \allowbreak \varphi, \allowbreak \psi,
\allowbreak e)$ returns on termination either a program~$P'$ correct
with respect to~$\varphi$ and~$\psi$ if $P$ can be repaired in $e$ or
otherwise throws an error.
\end{thm}

\begin{proof}
Since~$P$ is faulty with respect to the specification, model checking
in the first step produces a counterexample~$\pi$ on termination.  If
the model checking procedure does not terminate, then \PBRepair{} does
not terminate.  If the fault region~$e$ is not contained in $\pi$,
then \PBRepair{} terminates with an error indicating that~$P$ cannot
be repaired within the fault region~$e$.  This may happen when
multiple faults are considered.  Otherwise, a synthesis procedure is
invoked to repair the program~$P$ at~$e$.  If the procedure does not
terminate, \PBRepair does not terminate. If the synthesis procedure
reports unrealizability of the specification, no repair exists to make
$P$ correct in $e$ and \PBRepair{} throws an error.  Otherwise,
according to Lemma~\ref{ref:iter} program terms~$\overline{T}$ are
synthesized such that the counterexample is removed.  In each
iteration, at least one counterexample is removed.  Since the
verification conditions of the counterexamples are accumulated,
monotonicity is enforced, i.e., previously corrected counterexamples
cannot become faulty again.  Since~$P$ has only finitely many
control-flow paths, \PBRepair{} terminates after finitely many
iterations if for all iterations, model checking and synthesis
terminate.  The finally produced program has no counterexamples and
thus is correct with respect to $\varphi$ and $\psi$.
\end{proof}

\section{Implementation and Experimental Results}\label{sec:eval}
The repair framework presented in the previous section is generic in
the sense that the framework can be instantiated with different model
checkers and synthesis procedures abstracting from programming and
specification languages.  In this section, we present a prototype
implementation of the repair framework for ANSI-C utilizing domain
finitizing and give some initial experimental results indicating that
path-based program repair can be useful for repairing real programs.

The specification, i.e., the precondition and postcondition, and all
other logic formul\ae{} are expressed in the SMT-LIB2 logic QF\_BV,
i.e., the quantifier-free fragment of first-order logic modulo
bit-vector arithmetic.  The manipulation of logic formul\ae{} for
computing weakest preconditions and strongest postconditions has been
implemented using the API of the theorem prover Z3~\cite{Moura08}.

Symbolic counterexamples are computed by leveraging
CBMC~\cite{Clarke04} in combination with a self-im\-ple\-men\-ted
execution tracer.  CBMC model checks the program with respect to the
given specification.  When verification fails, an input assignment is
extracted from CBMC's logfile.  The execution tracer then re-simulates
the program with this input assignment and dumps the statements
executed in a textual representation similar to
\Fig~\ref{Fig:Propagation}.

For synthesis, logic formul\ae{} are bit-blasted to Boolean functions,
more particularly \emph{And-Inverter Graphs}~(AIGs), by replacing each
word-level variable by individual bit-level variables and each
bit-vector operator by a corresponding Boolean circuit.  After
bit-blasting, a BDD-based synthesis procedure is applied to obtain a
gate-level repair for the program.  The synthesis procedure is
complete, guarantees termination, and thus detects unrealizability.
Basing the main synthesis work on BDDs has many advantages --- they
can perform the quantifier elimination step needed for synthesis in a
natural and efficient way.  Also, the question of how to compute an
implementation from an input/output relation that is represented as a
Boolean function is well-researched, so that we can apply this work.

In the last step, the gate-level repair is transformed to ANSI-C code in
a straight-forward way: for each circuit gate~$o = AND(a,b)$, a fresh
variable~$o$ is introduced and assigned to the ANSI-C
expression~$(a~\&~b)$, where $a$ and $b$ are either other variables
introduced by this conversion or input bits extracted from existing
word-level program variables.

\Fig~\ref{Fig:minmax} shows {\tt minmax}, a simple fragment of an
ANSI-C program that determines the largest and the smallest value of
three given inputs.  All variables in the program fragment are of
integer type.  The logic specification $\varphi$ and $\psi$ annotated
to the source code is complete, so that all possible faults are
observable during model checking and can be repaired by our approach
assuming the ``right'' fault region is provided as input.  To improve
scalability of synthesis, the bit-widths of integer variables are
reduced to 2~bit.

\begin{figure}[t]
  \centering
  \scriptsize
  \begin{minipage}{0.58\columnwidth}
    \centering
    \begin{Verbatim}[commandchars=\\\{\},codes={\catcode`\$=3\catcode`\^=7\catcode`\_=8}]
      \PY{c+c1}{// $\mathbf{\varphi :\Leftrightarrow \{true\}}$}
      \PY{l+m+mf}{1.} \PY{n}{most} \PY{o}{=} \PY{n}{input1}\PY{p}{;}
      \PY{l+m+mf}{2.} \PY{n}{least} \PY{o}{=} \PY{n}{input1}\PY{p}{;}
      \PY{l+m+mf}{3.} \PY{k}{if} \PY{p}{(}\PY{n}{most} \PY{o}{\PYZlt{}} \PY{n}{input2}\PY{p}{)}
      \PY{l+m+mf}{4.}   \PY{n}{most} \PY{o}{=} \PY{n}{input2}\PY{p}{;}
      \PY{l+m+mf}{5.} \PY{k}{if} \PY{p}{(}\PY{n}{most} \PY{o}{\PYZlt{}} \PY{n}{input3}\PY{p}{)}
      \PY{l+m+mf}{6.}   \PY{n}{most} \PY{o}{=} \PY{n}{input3}\PY{p}{;}
      \PY{l+m+mf}{7.} \PY{k}{if} \PY{p}{(}\PY{n}{input2} \PY{o}{\PYZlt{}} \PY{n}{least}\PY{p}{)}
      \PY{l+m+mf}{8.}   \PY{n}{least} \PY{o}{=} \PY{n}{input2}\PY{p}{;}
      \PY{l+m+mf}{9.} \PY{k}{if} \PY{p}{(}\PY{n}{input3} \PY{o}{\PYZlt{}} \PY{n}{least}\PY{p}{)}
      \PY{l+m+mf}{10.}   \PY{n}{least} \PY{o}{=} \PY{n}{input3}\PY{p}{;}
      \PY{c+c1}{// $\mathbf{\psi :\Leftrightarrow \{}$}
      \PY{c+c1}{// $\mathbf{  (most=input1 \vee most=input2 \vee most=input3) \wedge}$}
      \PY{c+c1}{// $\mathbf{  (most\geq{}input1 \wedge most\geq{}input2 \wedge most\geq{}input3) \wedge}$}
      \PY{c+c1}{// $\mathbf{  (least=input1 \vee least=input2 \vee least=input3) \wedge}$}
      \PY{c+c1}{// $\mathbf{  (least\leq{}input1 \wedge least\leq{}input2 \wedge least\leq{}input3) \}}$}
    \end{Verbatim}
  \end{minipage}
  \caption{{\tt minmax} program}
  \label{Fig:minmax}
\end{figure}

In order to allow repairing conditional statements of form
\texttt{if(c)\{...\}}, where the guard condition~$c$ may be a complex
or compound expression, in a preprocessing step, the conditional
statement is replaced by \texttt{t = c; if(t)\{...\}}, where $t$ is a
new temporary program variable.

Table~\ref{Tab:PBRepair} lists some initial experiments, where faults
have been seeded into {\tt minmax} and \PBRepair{} is applied to
repair them.  The table is built as follows: each line corresponds to
one seeded fault.  The first column shows the line number in which a
fault was seeded, the second column lists the type of the erroneous
statement, the third column gives the number of iterations needed
by~\PBRepair{} to terminate and the fourth column lists the examined
control-flow paths as bit strings for all iterations and the size of
the corresponding candidate repair counted in AND gates in squared
brackets.  Each bit string $g_1 g_2 g_3 g_4$ denotes the evaluation of
the guard conditions, where $g_1$, $g_2$, $g_3$, $g_4$ correspond to
the code lines 3, 5, 7, 9, respectively.  The value $0$ and $1$
indicate that the respective guard condition evaluated to \FALSE and
\TRUE, respectively, when executed. The last column gives the run-time
in seconds.  All experiments have been conducted on Intel(R) Core(TM)
i5-2520M CPU @ 2.50GHz with 8GB RAM.  The run-time was mainly spend in
synthesizing the repair and the time required for model checking was
negligible.

\begin{table}[t]
  \centering
  \small
  \begin{tabular}{rcclc}
    \toprule
    Line & Type & Iterations & Control-Flow Path [AND gates] & Time [s]\\
    \midrule
     1 & Assignment & 7 & 0000[2] 0100[2] 0011[16] 0110[17] 1010[12] 0010[18] 0001[7] & 2 \\
     2 & Assignment & 4 & 1001[2] 0100[3] 0000[2] 1000[2] & 3\\
     3 & Condition  & 4 & 1110[1] 0000[3] 1010[6] 0001[8] & 4\\
     4 & Assignment & 2 & {1}000[3] {1}001[6] & 3\\
     5 & Condition  & 2 & 0101[1] 0000[5] & 6\\
     6 & Assignment & 2 & 0{1}10[3] 0{1}00[6] & 4\\
     7 & Condition  & 5 & 0001[1] 1010[3] 0100[9] 0000[11] 1110[8] & 4\\
     8 & Assignment & 3 & 01{1}1[2] 00{1}0[5] 01{1}0[4] & 4\\
     9 & Condition  & 6 & 0111[1] 0000[1] 1001[4] 0011[8] 1000[18] 0010[20] & 5\\
    10 & Assignment & 2 & 001{1}[2] 000{1}[14] & 4\\
    \bottomrule
  \end{tabular}
  \caption{Path-Based Program Repair applied to {\tt minmax}}
  \label{Tab:PBRepair}
\end{table}

The repair framework proposed is fully automated and does not need any
human intervention.  Our initial experiments indicate that \PBRepair{}
can be used for repairing simple ANSI-C programs; i.e., the prototype
implementation of \PBRepair{} proposed was able to determine a repair
for each of our seeded faults in only a few seconds.  However, before
applying our repair procedure to {\tt minmax}, the bit-width of
integers was manually reduced.  Otherwise, quantification on up to 160
BDD variables is necessary which is challenging for
today's BDD-based procedures.  We claim that automated bit-width
abstraction refinement for synthesis is in reach such that
a synthesize-and-generalize approach is possible: first bit-widths are
abstracted to a small number of BDD variables, a repair is synthesized
from the abstraction, generalized to the full bit-widths, and
subsequently verified considering the context of the program.

\section{Conclusion}\label{sec:conclusion}

In this paper, we presented a path-based abstraction refinement
approach to program repair which combines symbolic path reasoning and
software synthesis.  A prototype implementation of the repair
framework has been presented utilizing domain finitizing and BDD-based
synthesis.  In contrast to other synthesis approaches, this allows for
deciding realizability.  Initial experimental results for our
prototype on a small ANSI-C program have been presented.

The repair framework uses off-the-shelf model checking and synthesis
tools, and thus inherits their scalability strength and barriers.  In
case of the BDD-based synthesis the limiting factor is the number of
input and output variables after bit-blasting, which were manually
reduced for our experiments.  We conjecture that a customized
bit-width abstraction refinement approach will substantially improve
scalability, while allowing to keep the completeness and
unrealizability detecting capabilities of BDD-based synthesis.  A
challenge that remains in this context is to foster readability of the
computed implementation parts.  We leave these improvements to future
work.

\textbf{Acknowledgements:} This work was supported by the German
Research Foundation (DFG, grant no. FE 797/6-1) and the Institutional
Strategy of the University of Bremen, funded by the German Excellence
Initiative.

\bibliographystyle{eptcs}
\bibliography{ref}

\end{document}

%% file: mainmac.tex
\providecommand{\event}{FESCA 2015} 
\usepackage{breakurl}               

\newcommand{\Alg}{Alg.}
\newcommand{\Fig}{Fig.}
\newcommand{\Sec}{Sec.}
\newcommand{\Tab}{Tab.}

\usepackage{xspace}

\usepackage[ruled,linesnumbered]{algorithm2e}

\usepackage{graphicx}

\usepackage{booktabs}
\usepackage{tabularx}

\usepackage{stmaryrd}
\usepackage{amsmath}
\usepackage{amsthm}
\usepackage{amsfonts}
\theoremstyle{plain}
\newtheorem{thm}{Theorem}[section]
\newtheorem{ex}{Example}[section]
\newtheorem{defn}{Definition}[section]
\newtheorem{lem}[thm]{Lemma}

\newcommand{\TRUE}{\ensuremath{\mathbf{true}}\xspace}
\newcommand{\FALSE}{\ensuremath{\mathbf{false}}\xspace}
\newcommand{\FOL}{\ensuremath{\mathsf{FOL}}\xspace}
\newcommand{\FV}{\ensuremath{\mathop{FV}}\xspace}
\newcommand{\Vars}{\ensuremath{\mathop{Vars}}\xspace}
\newcommand{\BV}{\ensuremath{\mathsf{FOL}}\xspace}
\newcommand{\Stmts}{\ensuremath{\mathsf{Stmts}}\xspace}
\newcommand{\PBRepair}{\textsf{PBRepair}\xspace}

%% file: pictures.tex
\usepackage{fancyvrb}
\usepackage{color}
\input{pygments}

%% file: pygments.tex
\def\xskip{\hskip 7pt plus 3pt minus 4pt}

\newdimen\algindent
\newif\ifitempar \itempartrue
\def\algindentset#1{\setbox0\hbox{{\bf #1.\kern.25em}}\algindent=\wd0\relax}
\def\algbegin #1 #2{\algindentset{#21}\alg #1 #2}
\def\aalgbegin #1 #2{\algindentset{#211}\alg #1 #2}
\def\alg#1(#2). {\medbreak
  \noindent{\bf#1}{({\it#2\/})}.\xskip\ignorespaces}
\def\algstep#1.{\ifitempar\smallskip\noindent\else\itempartrue
  \hskip-\parindent\fi
  \hbox to\algindent{\bf\hfil #1.\kern.25em}%
  \hangindent=\algindent\hangafter=1\ignorespaces}
\def\slug{\hbox{\kern1.5pt\vrule width2.5pt height6pt depth1.5pt\kern1.5pt}}

\makeatletter
\def\PY@reset{\let\PY@it=\relax \let\PY@bf=\relax%
    \let\PY@ul=\relax \let\PY@tc=\relax%
    \let\PY@bc=\relax \let\PY@ff=\relax}
\def\PY@tok#1{\csname PY@tok@#1\endcsname}
\def\PY@toks#1+{\ifx\relax#1\empty\else%
    \PY@tok{#1}\expandafter\PY@toks\fi}
\def\PY@do#1{\PY@bc{\PY@tc{\PY@ul{%
    \PY@it{\PY@bf{\PY@ff{#1}}}}}}}
\def\PY#1#2{\PY@reset\PY@toks#1+\relax+\PY@do{#2}}

\expandafter\def\csname PY@tok@k\endcsname{\def\PY@tc##1{\textcolor[rgb]{0.00,0.50,0.00}{##1}}}
\expandafter\def\csname PY@tok@c\endcsname{\let\PY@it=\textit\def\PY@tc##1{\textcolor[rgb]{0.25,0.50,0.50}{##1}}}
\def\PYZlt{\char`\<}
\makeatother
